\documentclass[11pt]{article}
\pdfoutput=1
\usepackage[margin=1in]{geometry}
\usepackage[onehalfspacing]{setspace}
\usepackage{graphicx,amssymb}
\usepackage{amsmath,mathtools,amsthm}
\usepackage{float}
\usepackage{xfrac}
\usepackage{enumerate,multicol,multirow}
\usepackage{caption,subcaption}
\usepackage{hyperref}
\usepackage[capitalize]{cleveref}
\usepackage{tikz}
\usepackage{pgfplots}
\pgfplotsset{compat=1.18}

\usepackage{amssymb}

\usepackage{thm-restate}
\usepackage{natbib}
\usepackage{xcolor}
\usepackage{booktabs}
\usepackage{tabularx}

\usepackage{todonotes}

\usepackage{color-edits}
\addauthor{yl}{blue} % yl for Yingkai
\addauthor{jk}{red} %jk for Jussi

\theoremstyle{plain}
\newtheorem{definition}{Definition}

\newtheorem{theorem}{Theorem}

\newtheorem{proposition}{Proposition}

\newtheorem{assumption}{Assumption}
\crefname{assumption}{Assumption}{Assumptions}
\Crefname{assumption}{Assumption}{Assumptions}

\newcommand{\notshow}[1]{{}}

\newcommand{\Xcomment}[1]{{}}

\newcommand{\given}{\,\vert\,}

\newcommand{\prob}[2][]{\text{\bf Pr}\ifthenelse{\not\equal{}{#1}}{_{#1}}{}\!\left[{\def\givenn{\middle|}#2}\right]}
\newcommand{\expect}[2][]{\text{\bf E}\ifthenelse{\not\equal{}{#1}}{_{#1}}{}\!\left[{#2}\right]}

\newcommand{\tparen}{\big}
\newcommand{\tprob}[2][]{\text{\bf Pr}\ifthenelse{\not\equal{}{#1}}{_{#1}}{}\tparen[{\def\given{\tparen|}#2}\tparen]}
\newcommand{\texpect}[2][]{\text{\bf E}\ifthenelse{\not\equal{}{#1}}{_{#1}}{}\tparen[{\def\given{\tparen|}#2}\tparen]}

\newcommand{\sprob}[2][]{\text{\bf Pr}\ifthenelse{\not\equal{}{#1}}{_{#1}}{}[#2]}
\newcommand{\sexpect}[2][]{\text{\bf E}\ifthenelse{\not\equal{}{#1}}{_{#1}}{}[#2]}

\newcommand{\rbr}[1]{\left(#1\right)}

\newcommand{\signal}{\sigma}

\newcommand{\rev}{{\rm Rev}}
\newcommand{\opt}{{\rm OPT}}
\newcommand{\val}{v}

\title{Three Tiers and Thresholds: Incentives in Private Market Investing\thanks{Yingkai Li thanks the NUS Start-up Grant for financial support.}}

\author{Jussi Keppo\thanks{Department of Analytics \& Operations,
NUS Business School, National University of Singapore.
Email: \texttt{keppo@nus.edu.sg}.}
\and Yingkai Li\thanks{Department of Economics, National University of Singapore.
Email: \texttt{yk.li@nus.edu.sg}.} }

\date{}
\begin{document}

\maketitle

\begin{abstract}
This paper studies optimal contract design in private market investing, focusing on internal decision making in venture capital and private equity firms. A principal relies on an agent who privately exerts costly due diligence effort and then recommends whether to invest. Outcomes are observable ex post even when an opportunity is declined, allowing compensation to reward both successful investments and prudent decisions to pass. We characterize profit maximizing contracts that induce information acquisition and truthful reporting. We show that three tier contracts are sufficient, with payments contingent on the agent’s recommendation and the realized return. In symmetric environments satisfying the monotone likelihood ratio property, the optimal contract further simplifies to a threshold contract that pays only when the recommendation is aligned with an extreme realized return. These results provide guidance for performance based compensation that promotes diligent screening while limiting excessive risk taking.

\end{abstract}


% \noindent \textbf{JEL Codes:} 
\section{Introduction}
\label{sec:intro}
This paper studies optimal incentive design in private equity and venture capital firms, where senior partners or investment committees rely on junior investment professionals to evaluate uncertain investment opportunities. We consider a delegated decision making problem in which a principal (e.g., a senior partner or an investment committee) must incentivize an agent (e.g., a junior investment professional) to exert costly effort to screen an opportunity and to truthfully report her private recommendation. Motivated by the rapid growth of private equity and venture capital, and by the layered performance contingent compensation observed in practice, we show that optimal incentives take simple, structured forms that mirror common internal schemes. In particular, beyond fund level carried interest, promotion and bonus decisions often hinge on whether an investment clears a return hurdle and, importantly, on whether the agent also avoids weak deals by recommending to pass. Our model provides a theoretical foundation for these practices: the optimal three tier and threshold contracts balance incentives for effortful screening with incentives for truthful reporting, rationalizing why hurdle based rules remain effective inside modern investment organizations.

We model this setting as a principal agent problem with two key frictions: the agent's information gathering effort is unobservable, and her signal about the investment's potential return is private. The central design question is how to structure contracts that align incentives while balancing costs. A distinctive feature of our environment is that the outcome of a declined investment can be observed ex post. This reflects common scenarios in private markets where a competitor completes the deal, or the target company is later acquired, goes public, or raises subsequent funding rounds, all of which generate information about its eventual performance. This observability allows the principal to reward not only successful investments but also the prudent avoidance of poor ones. Moreover, in \cref{subapx:observe}, we show that this observability assumption can be relaxed with negligible profit loss for the principal. 

Our paper makes three contributions:
\begin{enumerate}
\item[$(i)$] We characterize optimal contracts in general environments with convex effort costs and finite return support, without requiring symmetry or the monotone likelihood ratio property (MLRP). We show that a three tier contract, where payments depend on both the agent's recommendation and the realized return, is sufficient.
\item[$(ii)$] Under additional structure (symmetry and MLRP), we show that the optimal contract simplifies to a (symmetric) threshold contract. Such contracts reward the agent only when her recommendation and the realized return cross specified thresholds. This form mirrors practice, for example hurdle based bonuses in VC firms.
\item[$(iii)$] We show that without MLRP, the optimal contract need not have a threshold form, and may exhibit a non threshold structure. In particular, when the signal does not satisfy MLRP, optimal incentives may require non threshold, potentially non monotonic payment schemes to preserve effort and reporting incentives (see Appendix B.2 for an example). More broadly, outside the symmetric benchmark, optimal contracts generally take asymmetric three tier forms that reflect heterogeneous beliefs and screening technologies.
\end{enumerate}

These results provide a theoretical foundation for layered performance based contracts widely used in PE and VC. While simple threshold contracts may be optimal under symmetry and MLRP, our general characterization highlights when more nuanced structures are necessary, especially when information is less well behaved. The analysis has implications for compensation design, hiring, and governance within investment organizations.

\subsection{Literature}
\label{subsec:literature}
This paper connects two major literatures: $(i)$ incentive design for information acquisition and truthful reporting, and $(ii)$ managerial incentives and contract design in finance and operations.

\paragraph{Elicitation and scoring rules.}
A central challenge in designing incentives for information production is eliciting truthful and informative reports from experts and forecasters. The foundation is proper scoring rules, which reward accurate probabilistic reports in expectation when outcomes are verifiable ex post \citep{Savage1971,GneitingRaftery2007}. Building on this methodology, \citet{Nau1985Effective} and \citet{Lichtendahl2007Competition} analyze the properties of scoring rules and competitive forecasting environments, while \citet{KilgourGerchak2004} study competitive rules that pay forecasters based on relative performance. Our paper instead considers a principal--agent setting without competition and designs mechanisms that maximize the principal's expected profit.

Our emphasis on incentivizing \emph{costly} information acquisition (rather than only truthful revelation of existing beliefs) is closely related to work on optimizing scoring rules to strengthen learning incentives. In computer science, \citet{li2022optimization} and \citet{HartlineEtAl2023COLT} develop frameworks for effort sensitive scoring rules in single and multi dimensional settings, and \citet{LiLibgober2023} extend these ideas to dynamic environments. More recently, \citet{wu2024incentivizing} studies continuous effort and shows that, under symmetric MLRP conditions, the optimal transfer takes a cutoff prize form similar to ours. The key difference lies in the objective: much of this literature maximizes forecast accuracy subject to an ex post budget constraint, whereas we maximize profit, the expected investment return net of payments. This shift implies that the optimal contract in our model may deliberately avoid inducing the most precise signal in order to limit commission payments. Moreover, relative to \citet{wu2024incentivizing}, we show that even without symmetry or MLRP, three tier contracts remain sufficient.

\paragraph{Managerial incentives and financial contracting.}
Incentive problems with unobservable effort are classical in economics \citep[e.g.,][]{Holmstrom1979,HolmstromMilgrom1987}. In entrepreneurial finance, venture contracts allocate cash flow and control rights in empirically rich ways \citep{KaplanStromberg2003,KaplanStromberg2004,EwensGorbenkoKorteweg2022}. At the partnership level, \citet{GompersLerner1999} document how compensation in U.S. venture capital limited partnerships varies across funds and over time. Relatedly, \citet{EwensRhodesKropf2015} show that partner level human capital is a first order determinant of VC performance, motivating why internal evaluation and incentive systems matter.

At the fund level, the structure and timing of private equity compensation, including management fees, carried interest, and hurdles, shape pay performance sensitivity and risk sharing \citep{MetrickYasuda2010}; see also \citet{Huther2020PE} for evidence linking contract features to realized outcomes. More broadly, private equity ownership can generate high powered incentives with real consequences for operating choices and outcomes \citep{GuptaHowellYannelisGupta2024}. Our setting connects these financial contracting and incentive based insights with elicitation design: an investment professional's information production is directly shaped by the mechanism used to evaluate and reward it.

A complementary stream studies how contracts transmit private forecasts and influence upstream or downstream decisions. \citet{CachonLariviere2001} examine how to share demand forecasts in supply chains through contractual design, and \citet{OzerWei2006} analyze how asymmetric forecast information interacts with capacity commitments. More broadly, organizational models of information acquisition and implementation (e.g., \citealp{ItohMorita2023}) explore who should learn, who should decide, and how incentives allocate attention and authority. Our analysis brings these themes to private market investing, where a junior investment professional's privately acquired information (or effort to obtain it) interacts with pay and verifiable performance in distinctive ways.

\section{Model}
\label{sec:model}
Our model is motivated by compensation design in private equity (PE) and venture capital (VC) firms, where junior investment professionals evaluate highly uncertain opportunities on behalf of senior partners or an investment committee. The framework captures two central frictions: an \emph{effort problem} (whether the junior investment professional conducts proper due diligence) and a \emph{reporting problem} (whether the junior investment professional truthfully reports her private recommendation).

The principal (the senior partner or the investment committee) faces a binary investment choice $\{I_0, I_1\}$.  
Option $I_0$ is a safe investment yielding a fixed, normalized return of $1$.  
Option $I_1$ is risky, with an uncertain return $\val \in \mathbb{R}_+$.  
There are two possible states for the risky option, $\{l, h\}$, where $l$ denotes the low state and $h$ denotes the high state.  
In state $l$, the return $\val$ is drawn from distribution $F_l$ with mean $\mu_l < 1$.  
In state $h$, $\val$ is drawn from distribution $F_h$ with mean $\mu_h > 1$.  
For tractability, both $F_l$ and $F_h$ are assumed to have finite support of size $n$, denoted $\{\val_1, \dots, \val_n\}$ with $\val_1 < \dots < \val_n$.  
Let $f_l \in \mathbb{R}_+^n$ and $f_h \in \mathbb{R}_+^n$ denote the corresponding probability mass functions, and let $p \in [0,1]$ be the prior probability of the high state $h$.

The agent (the junior investment professional) may exert costly effort to acquire a binary signal $\signal \in \{\signal_l, \signal_h\}$ about the underlying state. The signal is accurate with probability
\[
\gamma \triangleq \Pr[\signal_l | l] = \Pr[\signal_h | h] \geq \tfrac{1}{2},
\]
at a cost $c(\gamma)$.  
We assume $c(1/2) = 0$, so no cost is incurred without information acquisition.

\begin{assumption}[Convexity]\label{asp:convex}
The cost function $c(\gamma)$ is increasing and convex in accuracy $\gamma \in [\frac{1}{2},1]$.
\end{assumption}

The convexity of $c(\gamma)$ reflects the rising marginal effort required for more precise due diligence. Here, effort corresponds to the degree of signal accuracy $\gamma$, which can be interpreted as the depth of analysis. In practice, moving from a preliminary assessment (e.g., 60\%–70\% confidence) to a highly confident view (e.g., 80\%–90\%) typically requires disproportionately more time and resources. This captures a central tension in investment decision making: higher-quality analysis improves selection but benefits the analyst only indirectly unless incentives are aligned.

If the agent acquires information with an accuracy $\gamma \in [\frac{1}{2},1]$ and reports it truthfully, the principal's expected return is given by
\[
V(\gamma) = \Pr\nolimits_{\gamma}[\signal_h] \cdot \max\{1, \mathbb{E}_\gamma[\val|\signal_h]\} 
+ \Pr\nolimits_{\gamma}[\signal_l] \cdot \max\{1, \mathbb{E}_\gamma[\val|\signal_l]\}
\]
where $\Pr\nolimits_{\gamma}$ and $\mathbb{E}_\gamma$ indicate that the random signals have an accuracy of $\gamma$.

\subsection{Contracts}
The agent derives no intrinsic utility from information acquisition. To incentivize effort, the principal offers a contract.  
In our framework, the realized return $v$ of the risky option $I_1$ is observable to the principal ex post, regardless of whether the investment was made. This allows the contract to be contingent on the outcome even when the recommendation was to ``pass''. Such observability is plausible when, for example, a passed-on company is later funded by a competitor, or its valuation is revealed through a subsequent financing round or public market event.

More specifically, a contract $S$ is a mapping from the agent's report and realized return to a nonnegative payment:  
\[
S: \{\signal_l, \signal_h\} \times \mathbb{R} \rightarrow \mathbb{R}_+.
\]  
That is, $S(\signal,\val)$ specifies the payment to the agent given report $\signal$ and outcome $\val$.\footnote{Note that here we assume that the report space is identical to the signal space of the agent. This is without loss by the revelation principle \citep{myerson1981optimal}.}  
For convenience, let $r_l = S(\signal_l,\cdot)$ and $r_h = S(\signal_h,\cdot)$ denote the payment vectors conditional on reporting $\signal_l$ and $\signal_h$, and write $r=(r_l,r_h)$ for the full contract.

Given contract $r$, the total payment to the agent under accuracy $\gamma$ is given by
\begin{align*}
T(\gamma;r) = \max\Big\{ & (1-p)(\gamma f_l\cdot r_l + (1-\gamma)f_l \cdot r_h ) + p(\gamma f_h\cdot r_h + (1-\gamma)f_h \cdot r_l), \\
& (1-p)(\gamma f_l\cdot r_h + (1-\gamma)f_l \cdot r_l ) + p(\gamma f_h\cdot r_l + (1-\gamma)f_h \cdot r_h), \\
& (1-p)f_l \cdot r_l + p f_h \cdot r_l, \quad (1-p)f_l \cdot r_h + p f_h \cdot r_h \Big\},
\end{align*}
where the maximization reflects the possibility of misreporting.

The agent chooses accuracy to maximize utility
\[
U(\gamma;r) = T(\gamma;r) - c(\gamma),
\]
with optimal effort
\[
\gamma^*(r) = \arg\max_{\gamma} U(\gamma;r).
\]
The principal then selects $r$ to maximize expected net return:
\[
\rev(r) = V(\gamma^*(r)) - T(\gamma^*(r);r).
\]

\section{Optimal Contracts}
\label{sec:optimal}
We characterize optimal contracts that maximize the principal's payoff. 
The central tension is that the principal seeks strong incentives for the agent to acquire precise information, but providing such incentives requires high compensation when predictions are correct, which reduces net returns. 
The optimal contract therefore balances these forces, rewarding information acquisition while minimizing expected transfers.

We show that it suffices to restrict attention to contracts with three reward levels, based on the agent's recommendation and the realized return of the risky option.

\begin{definition}[Three tier contract]
\label{def:3-tier contract}
A contract $S$ is a three tier contract if there exist rewards $r_1,r_2\in\mathbb{R}_+$ such that 
$S(\signal,\val)\in \{0,r_1,r_2\}$ for all reports $\signal$ and realized returns $\val$.
\end{definition}

Before stating the formal result, it is useful to build intuition. 
The principal must both $(i)$ incentivize the agent to exert costly effort and $(ii)$ ensure truthful reporting of signals.  
Rewarding only ``invest \& high return'' would encourage exaggerated optimism, while rewarding only caution would encourage excessive pessimism. 
The cheapest way to achieve both goals is to reward aligned outcomes on both sides—``invest \& high return'' and ``pass \& low return''—and pay nothing elsewhere. 
This logic naturally yields a sparse three tier structure.

\begin{theorem}[Three tier optimality]\label{thm:optimal_contract}
Under \cref{asp:convex}, there exists an optimal contract that provides a three tier reward for the agent.
\end{theorem}

The three tier contract in Theorem~\ref{thm:optimal_contract} mirrors layered compensation schemes often observed in PE/VC firms. Junior and mid-level professionals commonly face incentives where $(i)$ failed or poor-return investments yield no reward, $(ii)$ modest success provides baseline recognition or promotion signals, and $(iii)$ high returns—especially when tied to bold or contrarian recommendations—unlock the highest bonuses or long-term incentives. Our model shows that such structures are not only practical but also theoretically sufficient for optimal incentive design. Moreover, the three tier form facilitates communication and buy-in, avoiding overly complex contingent contracts that are difficult to verify or explain internally.

The proof of \cref{thm:optimal_contract} (see \cref{apx:proofs}) proceeds in two steps. First, rather than directly optimizing over all contracts, we fix a target accuracy $\gamma^*>1/2$ and ask: among all contracts that induce $\gamma^*$, which minimizes expected payments? 
This \emph{payment-minimization} problem yields the cheapest implementing contracts for each $\gamma^*$, after which the principal selects the $\gamma^*$ that maximizes net value $V(\gamma^*)-T(\gamma^*;r)$. 
This two-step approach separates $(i)$ how to purchase a given level of screening accuracy at least cost from $(ii)$ how much accuracy is worth buying.

The reduced problem highlights two distinct forces shaping the contract:
\begin{enumerate}
\item \textbf{Effort incentives (moral hazard).} To induce $\gamma^*$, expected pay must be sufficiently \emph{sensitive} to reporting the correct signal: $T'(\gamma^*;r)=c'(\gamma^*)$. This sensitivity is generated by rewarding aligned outcomes. Stronger alignment pay steepens incentives but increases expected transfers (information rents).
\item \textbf{Reporting incentives (truthfulness).} If reward was placed only on one side—e.g., bonuses for ``invest \& high return''—then the agent observing bad signal would still prefer to recommend investing. Preventing this requires a second payoff region on the opposite alignment (``pass \& low return''). This additional tier is the cheapest way to deter misreporting without overpaying across states.
\end{enumerate}

The optimal contract balances these forces by concentrating payments on a small number of aligned outcomes and setting all others to zero. The payment-minimization problem is a linear program, and its extreme-point structure ensures that strictly positive rewards need only be placed on a few outcome–report cells. Hence a three tier contract—no pay, bonus type 1, bonus type 2—is sufficient. Economically, this sparse design supplies the marginal incentive power needed to implement the desired accuracy while avoiding unnecessary payments.

\subsection{Symmetric Environments with MLRP}
\label{subsec:symmetric-setup}
We next consider symmetric environments where the prior probability of the high state is $p=\tfrac{1}{2}$ and the probability mass functions $f_l,f_h$ are symmetric around a center point. 
Specifically, there exists $\hat{\val}$ such that for any $\delta\in\mathbb{R}$, we have
\[
f_l(\hat{\val}-\delta) = f_h(\hat{\val}+\delta).
\]
We also assume that the return distributions satisfy the monotone likelihood ratio property (MLRP):
\begin{assumption}[MLRP]\label{asp:mlrp}
The distributions $F_l$ and $F_h$ satisfy MLRP; that is, the likelihood ratio $\tfrac{f_h(\val)}{f_l(\val)}$ is weakly increasing in $\val$.
\end{assumption}

Under MLRP, higher returns are more indicative of state $h$ and lower returns of state $l$. 
This allows us to simplify the optimal contract further.

\begin{definition}[Threshold contract]
A contract $S$ is a threshold contract if there exist $r^*\in\mathbb{R}_+$ and fixed thresholds $\underline{\val},\bar{\val}\in \mathbb{R}_+$ such that 
\[
S(\signal, \val)=
\begin{cases}
r^* & \text{ if } (\signal = \signal_h \ \text{ and } \ \val \geq \bar{\val}) \ \text{or} \ (\signal = \signal_l \ \text{ and } \ \val \leq \underline{\val}),\\
0 & \text{otherwise}.
\end{cases}
\]
\end{definition}
Intuitively, under a threshold contract, payments occur only when the agent recommends investing and the realized return is sufficiently high, or when the agent recommends passing and the realized return is sufficiently low.

%\begin{definition}[Threshold contract]
%A contract $r$ is a threshold contract if there exist thresholds $\underline v \le \bar v$ such that
%\[
%r(\sigma_h, v) > 0 \ \text{only if } v \ge \bar v,
%\qquad
%r(\sigma_l, v) > 0 \ \text{only if } v \le \underline v,
%\]
%and $r(\sigma, v)=0$ for all other $(\sigma,v)$ cells.
%\end{definition}

\begin{theorem}[Threshold optimality]\label{thm:symmetric_mlrp}
Under \cref{asp:mlrp,asp:convex}, in symmetric environments there exists an optimal contract that is a threshold contract.
\end{theorem}

Threshold contracts closely resemble hurdle-based performance evaluation in practice. For instance, the agent may be credited for recommending a deal only if $(i)$ the recommendation was to invest, and $(ii)$ the realized IRR exceeds a predefined hurdle (e.g., 20\%). Similarly, for conservative calls, the agent may be rewarded if declining an investment avoids underperformance. This dual alignment—rewarding both correct enthusiasm and correct caution—is precisely what the threshold contract captures.

In symmetric environments with MLRP, \cref{thm:symmetric_mlrp} shows that the cheapest way to buy accuracy is to pay only when the realized outcome is most \emph{diagnostic} of the state and matches the recommendation. Since the likelihood ratio is monotone, upper-tail outcomes are the strongest evidence for the high state and lower-tail outcomes for the low state. Paying only in these tails maximizes incentive power per dollar spent: each rewarded observation provides the greatest increase in expected payoff conditional on truthful reporting. Symmetry then aligns the thresholds and allows the principal to use a single bonus across both sides, preserving truthfulness without unnecessary transfers.

This result has practical appeal. Threshold contracts are transparent, easy to implement, and reduce the risk of discouraging bold or contrarian recommendations. The symmetry assumption reflects balanced ex ante beliefs about success and failure, while MLRP captures the idea that stronger signals are more likely in better states. 

A key nuance is that raising rewards does not always improve the principal's payoff: while higher accuracy avoids mistakes, it also increases expected payments. Hence the optimal contract balances the marginal value of accuracy against its cost. This insight is particularly relevant in PE/VC, where carry-based compensation and promotion tracks must be calibrated carefully.

When symmetry fails, optimal contracts may take an asymmetric threshold form. When MLRP fails, optimal contracts can require genuinely non threshold, possibly non monotonic structures (Appendix~\ref{subapx:non-threshold}). Our general result (Theorem~\ref{thm:optimal_contract}) continues to apply, but the implied payment schedule need not be a simple threshold rule. In such cases, firms may favor simpler designs for implementability, even at some cost to optimality. Empirical work could examine how compensation structures vary across settings with different degrees of informational asymmetry and signal regularity.

\section{Numerical Example}
\label{sec:numerical}
To illustrate the performance gains from our contract structure, we compare it with a simple linear contract. Under a linear scheme, the agent's compensation is a fixed share of realized returns, independent of her recommendations. Such contracts are common benchmarks in practice because they are easy to compute and administer.

We construct a stylized example with two states ($l,h$), equal priors ($p=1/2$), and finite support returns ($v\in\{0,1,2\}$). The return distributions $F_l$ and $F_h$ take the following form. 
\begin{align*}
F_l(v) = \begin{cases}
\frac{3}{5} & v=0 \\
\frac{1}{5} & v=1 \\
\frac{1}{5} & v=2, 
\end{cases}
\qquad
F_h(v) = \begin{cases}
\frac{1}{5} & v=0 \\
\frac{1}{5} & v=1 \\
\frac{3}{5} & v=2. 
\end{cases}
\end{align*}

In the low state $l$, the distribution $F_l$ has a mean $\mu_l=0.6<1$, while in the high state $h$, the distribution $F_h$ has a mean $\mu_h=1.4>1$. The agent chooses signal accuracy $\gamma\in[0.5,1]$ at convex cost $c(\gamma)=k(\gamma-0.5)^2$. The principal invests if the posterior expected return exceeds~1.

\paragraph{Contracts.}  
$(i)$ \emph{Linear share}: $S(v)=\alpha\cdot v$, independent of the agent's report.  
$(ii)$ \emph{Threshold}: pay $r>0$ only when the report and realized return are aligned and extreme, i.e., ``invest \& high return" (report $h$ and $v=2$) or ``pass \& low return" (report $l$ and $v=0$).  

\paragraph{Results.}  
Optimizing contract parameters for $k=\frac{1}{15}$ yields the following comparison:

\begin{table}[ht]
\caption{Comparison of optimized linear vs.\ threshold contracts ($k=\frac{1}{15}$).}
\label{tab:contract-comparison}
\centering
\small
\begin{tabular}{lccccc}
\toprule
Contract & Parameters & $\gamma^*$ & Principal value $V(\gamma^*)$ & Payment $T(\gamma^*)$ & Net payoff $V-T$ \\
\midrule
Linear share & $\alpha=\frac{1}{12}$ & 0.75 & 1.1 & 0.092 & 1.008 \\
Threshold & $r=\frac{1}{6}$ & 1.00 & 1.2 & 0.1 & 1.1 \\
\bottomrule
\end{tabular}
\end{table}

\begin{figure}[t]
\centering
\begin{tikzpicture}
\begin{axis}[
  width=13cm, height=7.5cm,
  xlabel={$k$}, ylabel={Investment Return},
  xmin=0, xmax=0.25,
  ymin=1, ymax=1.21,
  xtick={0,0.1,0.2,0.25},
  ytick={1,1.05,1.10,1.15,1.20},
  grid=major,
  legend style={draw=none, fill=none, font=\small},
  legend pos=north east,
  unbounded coords=jump,
  domain=0:0.25,
]

% ----- thresholds -----
\pgfmathsetmacro{\kPone}{0.1}
\pgfmathsetmacro{\kPtwo}{0.2}
\pgfmathsetmacro{\kQone}{2/35} % ≈ 0.057142857
\pgfmathsetmacro{\kQtwo}{2/25} % = 0.08

% ===== P*(k): original P(r) with r >= 0  (solid blue) =====
\addplot[very thick, blue, forget plot] expression[domain=0:\kPone] {1.2 - 1.5*x};
\addplot[very thick, blue, forget plot] expression[domain=\kPone:\kPtwo, samples=200] {0.8 + 0.5*x + 0.02/x};
\addplot[very thick, blue, forget plot] expression[domain=\kPtwo:0.25] {1};

% ===== Q*(k): corrected Q(a)  (red dotted) =====
\addplot[very thick, red, dotted, forget plot] expression[domain=0:\kQone] {1.2 - 3*x};
\addplot[very thick, red, dotted, forget plot] expression[domain=\kQone:\kQtwo, samples=200] {0.5 + 0.02/x + (25.0/8.0)*x};
\addplot[very thick, red, dotted, forget plot] expression[domain=\kQtwo:0.25] {1};

% ----- explicit legend images to match styles -----
\addlegendimage{line legend, red, dotted, very thick}
\addlegendentry{Linear Contracts}
\addlegendimage{line legend, blue, very thick}
\addlegendentry{Threshold Contracts}

% (Optional) vertical guides at breakpoints
\addplot[densely dotted, gray, forget plot] coordinates {(\kQone,1) (\kQone,1.21)};
\addplot[densely dotted, gray, forget plot] coordinates {(\kQtwo,1) (\kQtwo,1.21)};
\addplot[densely dotted, gray, forget plot] coordinates {(\kPone,1) (\kPone,1.21)};
\addplot[densely dotted, gray, forget plot] coordinates {(\kPtwo,1) (\kPtwo,1.21)};

\end{axis}
\end{tikzpicture}
\caption{The expected payoff from optimal investments as a function of $k$.}
    \label{fig:example}
\end{figure}
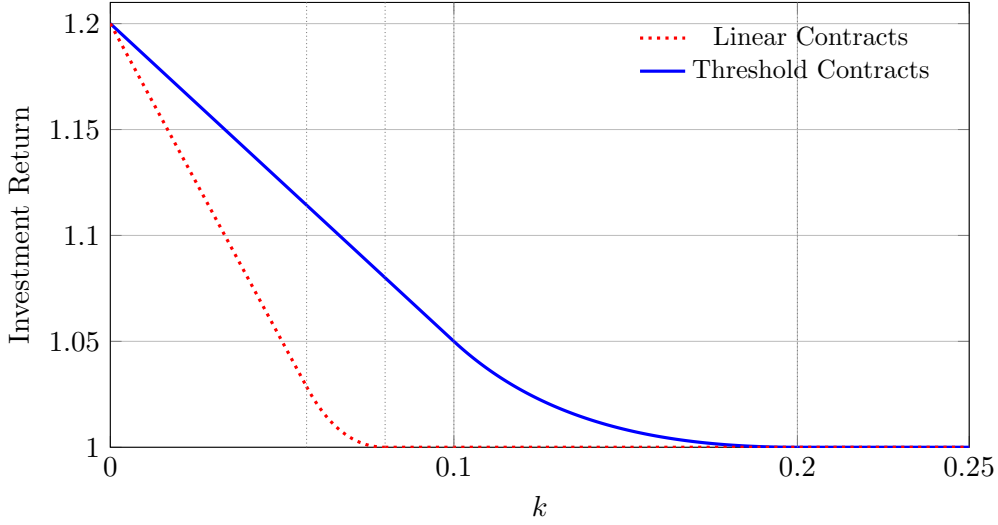

The threshold contract achieves higher induced accuracy ($\gamma^*=1.00$ vs.\ $0.75$) and raises the principal's net payoff by roughly 9\% (1.1 vs.\ 1.008).  

\paragraph{Interpretation.}  
The linear contract spreads transfers across all outcomes, including those that are only weakly informative, which is an expensive way to buy effort. The threshold contract concentrates payments on tail events that are most diagnostic of the true state and aligned with the agent's recommendation, which maximizes incentive power per dollar.  

The previous example serves as an illustration of the optimal contract and its comparison with linear contracts. 
In this example, the multiplicative gap between linear contracts and the optimal is roughly $1.1$. We show that in any environment, the gap is bounded, which is at most $2$.
The proof of \cref{prop:linear_gap} is provided in \cref{subapx:linear_gap}
\begin{proposition}\label{prop:linear_gap}
The multiplicative revenue gap between optimal contracts and linear contracts is at most $2$. 
\end{proposition}

\section{Conclusion}
This paper develops theoretical foundations and practical guidance for incentive design in high-stakes investment settings such as private equity (PE) and venture capital (VC) firms. Our results show that relatively simple performance-contingent contracts—especially those with a three tier reward structure—can align junior professionals' incentives with firm-level investment goals.

The theoretical justification for these contracts is strongest under symmetric beliefs and monotone signal informativeness. In such environments, threshold contracts not only achieve incentive alignment but are also straightforward to implement and communicate. More generally, even without symmetry or MLRP, three tier contracts remain sufficient for optimality, balancing informativeness with contract simplicity.

% Although we assume that both reporting and effort are unobservable, the framework can accommodate natural extensions, including partial observability, heterogeneous agent types, or multi-agent settings. Likewise, while we focus on binary recommendations, the structure can be generalized to richer reporting formats or contexts in which the agent has decision rights.

Taken together, the analysis provides a tractable benchmark for contract design in investment organizations where decisions are delegated under uncertainty and professionals hold private, costly-to-acquire information. The framework also yields testable predictions: firms operating in more symmetric deal environments (e.g., mid-stage venture capital) are expected to adopt simpler hurdle-based contracts, whereas firms facing highly idiosyncratic or distressed opportunities may rely on more complex or asymmetric designs. Empirical work comparing contract forms across sectors could provide direct evidence for these predictions and inform the evolution of compensation practices.

\paragraph{Managerial implications.}
For practitioners, the results suggest that firms need not design overly complex incentive schemes to achieve alignment. A small number of transparent reward tiers can be sufficient to motivate diligence and truthful reporting, while also being easier to communicate internally and to enforce. Threshold-based schemes, in particular, provide clear performance benchmarks that resonate with both analysts and senior partners. By calibrating these thresholds to the informational environment, firms can balance the value of better screening against the costs of additional compensation, thereby improving both decision quality and organizational cohesion.

\newpage

\bibliographystyle{apalike}
\bibliography{ref}

@inproceedings{li2022optimization,
  title={Optimization of scoring rules},
  author={Li, Yingkai and Hartline, Jason D and Shan, Liren and Wu, Yifan},
  booktitle={Proceedings of the 23rd ACM Conference on Economics and Computation},
  pages={988--989},
  year={2022}
}

@article{wu2024incentivizing,
  title={Incentivizing Information Acquisition},
  author={Wu, Fan},
  journal={arXiv preprint arXiv:2410.13978},
  year={2024}
}

@article{Lichtendahl2007Competition,
  title={Probability elicitation, scoring rules, and competition among forecasters},
  author={Lichtendahl Jr, Kenneth C and Winkler, Robert L},
  journal={Management Science},
  volume={53},
  number={11},
  pages={1745--1755},
  year={2007},
  publisher={INFORMS}
}

@article{Nau1985Effective,
  title = {Should Scoring Rules Be {}Effective{''}?},
  author={Nau, Robert F},
  journal={Management Science},
  volume={31},
  number={5},
  pages={527--535},
  year={1985},
  publisher={INFORMS}
}

@article{KilgourGerchak2004,
  title={Elicitation of probabilities using competitive scoring rules},
  author={Kilgour, D Marc and Gerchak, Yigal},
  journal={Decision Analysis},
  volume={1},
  number={2},
  pages={108--113},
  year={2004},
  publisher={INFORMS}
}

@article{CachonLariviere2001,
  title={Contracting to assure supply: How to share demand forecasts in a supply chain},
  author={Cachon, G{\'e}rard P and Lariviere, Martin A},
  journal={Management science},
  volume={47},
  number={5},
  pages={629--646},
  year={2001},
  publisher={Informs}
}

@article{OzerWei2006,
  title={Strategic commitments for an optimal capacity decision under asymmetric forecast information},
  author={{\"O}zer, {\"O}zalp and Wei, Wei},
  journal={Management science},
  volume={52},
  number={8},
  pages={1238--1257},
  year={2006},
  publisher={INFORMS}
}

@article{ItohMorita2023,
  title={Information acquisition, decision making, and implementation in organizations},
  author={Itoh, Hideshi and Morita, Kimiyuki},
  journal={Management Science},
  volume={69},
  number={1},
  pages={446--463},
  year={2023},
  publisher={INFORMS}
}

@article{Huther2020PE,
  title={Paying for performance in private equity: evidence from venture capital partnerships},
  author={H{\"u}ther, Niklas and Robinson, David T and Sievers, S{\"o}nke and Hartmann-Wendels, Thomas},
  journal={Management Science},
  volume={66},
  number={4},
  pages={1756--1782},
  year={2020},
  publisher={INFORMS}
}

@article{Holmstrom1979,
  title={Moral hazard and observability},
  author={Holmstr{\"o}m, Bengt},
  journal={The Bell journal of economics},
  pages={74--91},
  year={1979},
  publisher={JSTOR}
}

@article{HolmstromMilgrom1987,
  title={Aggregation and linearity in the provision of intertemporal incentives},
  author={Holmstr{\"o}m, Bengt and Milgrom, Paul},
  journal={Econometrica: Journal of the Econometric Society},
  pages={303--328},
  year={1987},
  publisher={JSTOR}
}

@article{KaplanStromberg2003,
  title={Financial contracting theory meets the real world: An empirical analysis of venture capital contracts},
  author={Kaplan, Steven N and Str{\"o}mberg, Per},
  journal={The review of economic studies},
  volume={70},
  number={2},
  pages={281--315},
  year={2003},
  publisher={Wiley-Blackwell}
}

@article{KaplanStromberg2004,
  title={Characteristics, contracts, and actions: Evidence from venture capitalist analyses},
  author={Kaplan, Steven N and Str{\"o}mberg, Per ER},
  journal={The journal of finance},
  volume={59},
  number={5},
  pages={2177--2210},
  year={2004},
  publisher={Wiley Online Library}
}

@article{MetrickYasuda2010,
  title={The economics of private equity funds},
  author={Metrick, Andrew and Yasuda, Ayako},
  journal={The Review of Financial Studies},
  volume={23},
  number={6},
  pages={2303--2341},
  year={2010},
  publisher={Oxford University Press}
}

@article{GneitingRaftery2007,
  author = {Gneiting, Tilmann and Raftery, Adrian E.},
  title = {Strictly Proper Scoring Rules, Prediction, and Estimation},
  journal = {Journal of the American Statistical Association},
  year = {2007},
  volume = {102},
  number = {477},
  pages = {359--378}
}

@article{Savage1971,
  author = {Savage, Leonard J.},
  title = {Elicitation of Personal Probabilities and Expectations},
  journal = {Journal of the American Statistical Association},
  year = {1971},
  volume = {66},
  number = {336},
  pages = {783--801}
}

@inproceedings{HartlineEtAl2023COLT,
  title={Optimal scoring rules for multi-dimensional effort},
  author={Hartline, Jason D and Shan, Liren and Li, Yingkai and Wu, Yifan},
  booktitle={The Thirty Sixth Annual Conference on Learning Theory},
  pages={2624--2650},
  year={2023},
  organization={PMLR}
}

@misc{LiLibgober2023,
  author = {Li, Yingkai and Libgober, Jonathan},
  title = {Implementing Evidence Acquisition: Time Dependence in Contracts for Advice},
  year = {2023},
  eprint = {2310.19147},
  archivePrefix = {arXiv}
}

@article{myerson1981optimal,
  title={Optimal auction design},
  author={Myerson, Roger B},
  journal={Mathematics of operations research},
  volume={6},
  number={1},
  pages={58--73},
  year={1981},
  publisher={INFORMS}
}

@article{GuptaHowellYannelisGupta2024,
  title={Owner incentives and performance in healthcare: private equity investment in nursing homes},
  author={Gupta, Atul and Howell, Sabrina T and Yannelis, Constantine and Gupta, Abhinav},
  journal={The Review of Financial Studies},
  volume={37},
  number={4},
  pages={1029--1077},
  year={2024},
  publisher={Oxford University Press}
}

@article{EwensGorbenkoKorteweg2022,
  author  = {Ewens, Michael and Gorbenko, Alexander and Korteweg, Arthur},
  title   = {Venture Capital Contracts},
  journal = {Journal of Financial Economics},
  year    = {2022},
  volume  = {143},
  number  = {1},
  pages   = {131--158},
  doi     = {10.1016/j.jfineco.2021.06.042}
}

@article{EwensRhodesKropf2015,
  author  = {Ewens, Michael and Rhodes-Kropf, Matthew},
  title   = {Is a VC Partnership Greater Than the Sum of Its Partners?},
  journal = {The Journal of Finance},
  year    = {2015},
  volume  = {70},
  number  = {3},
  pages   = {1081--1113},
  doi     = {10.1111/jofi.12249}
}

@article{GompersLerner1999,
  author  = {Gompers, Paul and Lerner, Josh},
  title   = {An Analysis of Compensation in the U.S. Venture Capital Partnership},
  journal = {Journal of Financial Economics},
  year    = {1999},
  volume  = {51},
  number  = {1},
  pages   = {3--44},
  doi     = {10.1016/S0304-405X(98)00042-7}
}
\newpage

\appendix
\section{Proofs of Optimal Contracts}
\label{apx:proofs}

\begin{proof}[Proof of \cref{thm:optimal_contract}]
Recall that given accuracy $\gamma$, the expected payment to the agent under contract $r$ is  
\begin{align*}
T(\gamma;r) &= \max\{(1-p)(\gamma f_l\cdot r_l + (1-\gamma)f_l \cdot r_h ) + p\cdot (\gamma f_h\cdot r_h + (1-\gamma)f_h \cdot r_l ), \\
&\qquad \qquad (1-p)(\gamma f_l\cdot r_h + (1-\gamma)f_l \cdot r_l ) + p\cdot (\gamma f_h\cdot r_l + (1-\gamma)f_h \cdot r_h ), \\
&\qquad \qquad (1-p)f_l \cdot r_l + p f_h\cdot r_l, 
\quad (1-p)f_l \cdot r_h + p f_h\cdot r_h\}.
\end{align*}
Note that it is without loss to consider the case where 
$f_l\cdot r_l \geq f_l \cdot r_h$
and $f_h\cdot r_h \geq f_h \cdot r_l$.
This is because otherwise either there is a dominant strategy for reporting the signals, which provides no incentives for the agent to exert effort, 
or we can swap the definition of $r_l$ and $r_h$ such that the above inequalities hold. 
In this case, the second case in $T(\gamma;r)$ can never be the maximum given any $\gamma$. Thus, the expected payment simplifies to 
\begin{align*}
T(\gamma;r) &= \max\{(1-p)(\gamma f_l\cdot r_l + (1-\gamma)f_l \cdot r_h ) + p\cdot (\gamma f_h\cdot r_h + (1-\gamma)f_h \cdot r_l ), \\
&\qquad \qquad (1-p)f_l \cdot r_l + p f_h\cdot r_l, 
\quad (1-p)f_l \cdot r_h + p f_h\cdot r_h\}.
\end{align*} 
Note that only the first case in $T(\gamma;r)$ depends on $\gamma$, and is linear in $\gamma$. 
Therefore, given any $r$, $T(\gamma;r)$ is a piecewise linear function in $\gamma$. 
Since $c(\gamma)$ is convex in $\gamma$, 
the utility maximizing accuracy choice for the agent is either $\gamma = \frac{1}{2}$, which has the lowest cost of $0$, 
or chooses a $\gamma$ that is high enough such that 
$$T(\gamma;r) = \hat{T}(\gamma;r) \triangleq (1-p)(\gamma f_l\cdot r_l + (1-\gamma)f_l \cdot r_h ) + p\cdot (\gamma f_h\cdot r_h + (1-\gamma)f_h \cdot r_l ).$$
In the former case, the principal makes the decision with only ex ante information without compensating the agent. This can be implemented as a zero contract, which is trivially a three tier contract.

In the latter case, the utility of the agent is $U(\gamma;r)=\hat{T}(\gamma;r) - c(\gamma)$, which is concave in $\gamma$ since $\hat{T}$ is linear in $\gamma$ and $c$ is convex in $\gamma$.
Therefore, the optimal choice of the agent in this case is determined by the first-order condition. 
That is, 
\begin{align*}
c'(\gamma) = (1-p)(f_l\cdot r_l -f_l \cdot r_h ) + p\cdot (f_h\cdot r_h - f_h \cdot r_l ).
\end{align*}
Suppose the revenue-optimal contract is $r^*$. 
Let 
\begin{align*}
d^* = (1-p)(f_l\cdot r^*_l -f_l \cdot r^*_h ) + p\cdot (f_h\cdot r^*_h - f_h \cdot r^*_l )
\end{align*}
and $\gamma^*$ be the accuracy such that $c'(\gamma^*) = d^*$. 
The optimal contract design can be reduced to designing a contract that incentivizes the agent to choose accuracy $\gamma^*$ while minimizing the expected payment. 
\begin{align}
\min_{r\geq 0}\quad& \hat{T}(\gamma^*;r) \tag{OPT}\label{eq:opt_contract_simplified}\\
\text{s.t.}\quad& (1-p)(f_l\cdot r_l -f_l \cdot r_h ) + p\cdot (f_h\cdot r_h - f_h \cdot r_l ) = d^* \tag{MH}\label{eq:MH}\\
& \hat{T}(\gamma^*;r) - c(\gamma^*)\geq (1-p)f_l \cdot r_l + p f_h\cdot r_l \tag{IC1}\label{eq:IC1}\\
&\hat{T}(\gamma^*;r) - c(\gamma^*)\geq (1-p)f_l \cdot r_h + p f_h\cdot r_h \tag{IC2}\label{eq:IC2}
\end{align}
A few comments on the constraints of the optimization program. 
\cref{eq:MH} ensures that the agent's local best response for accuracy is $\gamma^*$. 
Moreover, \cref{eq:IC1} and \eqref{eq:IC2} ensure that the agent prefers accuracy $\gamma^*$ and truthfully reporting the signals rather than shirking without effort. 
These two constraints together also imply that the agent has a weak incentive to report truthfully given accuracy $\gamma^*$, 
i.e., $T(\gamma^*;r) = \hat{T}(\gamma^*;r)$. 

\paragraph{Linear Programming Analysis:} Note that \eqref{eq:opt_contract_simplified} is a finite-dimensional LP. To determine the structure of the optimal solution, we convert it into standard equality form by introducing non-negative slack variables $s_1 \ge 0$ and $s_2 \ge 0$ for \eqref{eq:IC1} and \eqref{eq:IC2}, respectively. The constraints become:
\begin{align}
(1-p)(f_l\cdot r_l -f_l \cdot r_h ) + p\cdot (f_h\cdot r_h - f_h \cdot r_l ) &= d^*  \tag{MH} \\
\hat{T}(\gamma^{*};r) - ((1-p)f_l \cdot r_l + p f_h\cdot r_l) - s_1 &= c(\gamma^{*}) \tag{IC1'} \label{eq:ic1p}\\
\hat{T}(\gamma^{*};r) - ((1-p)f_l \cdot r_h + p f_h\cdot r_h) - s_2 &= c(\gamma^{*}) \tag{IC2'} \label{eq:ic2p}
\end{align}
This system has $M=3$ equality constraints. The total number of variables is $K=2n+2$ (the $2n$ payment variables $r_{\sigma,i}$ and the 2 slack variables $s_1, s_2$).

We invoke the Fundamental Theorem of Linear Programming, which states that if an LP has an optimal solution, then there exists an optimal solution that is a Basic Feasible Solution (BFS). A key property of a BFS in a system with $M$ equality constraints is that at most $M$ variables are non-zero.
Let $N_r$ be the number of positive payment variables in the optimal BFS, and $N_s$ be the number of positive slack variables ($N_s \in \{0, 1, 2\}$). We must have:
$$N_r + N_s \le 3.$$
We analyze the structure of the optimal contract by considering the possible values of $N_s$.

\paragraph{Case 1: At least one IC constraint is slack ($N_s \ge 1$).}
If $N_s \ge 1$ (i.e., $s_1 > 0$ or $s_2 > 0$), the number of positive payments is bounded by:
$$N_r \le 3 - N_s \le 3 - 1 = 2.$$
The optimal contract uses at most two positive payments.

\paragraph{Case 2: Both IC constraints are binding ($N_s = 0$).}
Let $R_l(r) = (1-p)f_l \cdot r_l + p f_h\cdot r_l$ and $R_h(r) = (1-p)f_l \cdot r_h + p f_h\cdot r_h$ be the expected payments from shirking. 
Let $d(r)$ be the slope of $\hat{T}(\gamma;r)$ with respect to $\gamma$:
\begin{align*}
d(r) \triangleq (1-p)(f_l\cdot r_l -f_l \cdot r_h ) + p\cdot (f_h\cdot r_h - f_h \cdot r_l ).
\end{align*}

If $N_s = 0$ ($s_1 = 0$ and $s_2 = 0$), both \eqref{eq:IC1} and \eqref{eq:IC2} hold with equality. 
This implies $R_l(r) = R_h(r)$. We examine the implications and consistency of these binding constraints.

% Let $R_l(r) = (1-p)f_l \cdot r_l + p f_h\cdot r_l$ and $R_h(r) = (1-p)f_l \cdot r_h + p f_h\cdot r_h$ be the expected payments from shirking.
% This implies that
% \begin{align*}
% (1-p)f_l \cdot r_l + p f_h\cdot r_l = (1-p)f_l \cdot r_h + p f_h\cdot r_h.
% \end{align*}
% Rearranging this equality gives:
% \begin{align}\label{eq:binding_condition}
% (1-p)(f_l\cdot r_l - f_l\cdot r_h) = p(f_h\cdot r_h - f_h\cdot r_l).
% \end{align}
% Let $D_l = f_l\cdot r_l - f_l\cdot r_h$ and $D_h = f_h\cdot r_h - f_h\cdot r_l$. The condition \eqref{eq:binding_condition} becomes $(1-p)D_l = p D_h$.

% If $N_s = 0$, (IC1') and (IC2') hold with equality. This implies $R_l(r) = R_h(r)$. We examine the implications and consistency of these binding constraints.

Due to the linearity of $\hat{T}(\gamma;r)$ in $\gamma$, we can decompose it as:
$$
\hat{T}(\gamma;r) = \hat{T}(\frac{1}{2};r) + (\gamma - \frac{1}{2}) \cdot d(r).
$$
Furthermore, $\hat{T}(\frac{1}{2};r) = \frac{1}{2} (R_l(r) + R_h(r))$.
In Case 2, since $R_l(r) = R_h(r)$, we have $\hat{T}(\frac{1}{2};r) = R_l(r)$.
Substituting this and the (MH) constraint $d(r)=d^*$ into the decomposition at $\gamma^*$:
\begin{align}\label{eq:T_identity}
\hat{T}(\gamma^*;r) = R_l(r) + (\gamma^* - \frac{1}{2}) \cdot d^*.
\end{align}
This reveals that if a contract $r$ satisfies $R_l(r)=R_h(r)$ and $d(r)=d^*$, the relationship between $\hat{T}(\gamma^*;r)$ and $R_l(r)$ is fixed:
\begin{align}\label{eq:IC_fixed_value}
\hat{T}(\gamma^*;r) - R_l(r) = (\gamma^* - \frac{1}{2}) \cdot d^*.
\end{align}

However, Case 2 also requires that \eqref{eq:IC1} holds with equality:
\begin{align}\label{eq:IC1_binding}
\hat{T}(\gamma^*;r) - R_l(r) = c(\gamma^*).
\end{align}
Comparing \eqref{eq:IC_fixed_value} and \eqref{eq:IC1_binding}, 
we see that a feasible solution for Case 2 exists if and only if
\begin{align}\label{eq:consistency}
(\gamma^* - 1/2) \cdot d^* = c(\gamma^*).
\end{align}
Since $c$ is convex with $c(1/2)=0$, we have 
$c'(\gamma^*) \ge \frac{c(\gamma^*)}{\gamma^*-1/2}$, 
so $(\gamma^* - 1/2)d^* \ge c(\gamma^*)$.

\subparagraph{Case 2a: Consistency condition holds.}
If $(\gamma^* - 1/2) d^* = c(\gamma^*)$ (i.e., the cost function is linear up to $\gamma^*$), Case 2 is feasible. Crucially, in this scenario, the constraints \eqref{eq:ic1p} and \eqref{eq:ic2p} are redundant. They are automatically satisfied by any contract $r$ that satisfies $d(r)=d^*$ and $R_l(r)=R_h(r)$, as demonstrated by \eqref{eq:IC_fixed_value} and \eqref{eq:consistency}.

We can thus formulate a reduced optimization problem using only the essential constraints. The objective $\hat{T}(\gamma^*;r)$ can be rewritten using \eqref{eq:T_identity} as $R_l(r) + (\gamma^* - 1/2) d^*$. Since the second term is constant, minimizing $\hat{T}(\gamma^*;r)$ is equivalent to minimizing $R_l(r)$.

The constraints $d(r)=d^*$ and $R_l(r)=R_h(r)$ can be combined algebraically. 
Define the expected payment gaps
\[
\Delta_l \triangleq f_l \cdot (r_l - r_h)
\quad\text{and}\quad
\Delta_h \triangleq f_h \cdot (r_h - r_l).
\]
Then $R_l(r)=R_h(r)$ is equivalent to $(1-p)\Delta_l = p\Delta_h$.
The moral hazard constraint is $(1-p)\Delta_l + p\Delta_h = d^*$.
%Let $D_l = f_l\cdot r_l - f_l\cdot r_h$ and $D_h = f_h\cdot r_h - f_h\cdot r_l$. The condition $R_l(r)=R_h(r)$ implies $(1-p)D_l = p D_h$. The moral hazard constraint is $(1-p)D_l + p D_h = d^*$. 
Combining these yields:
\begin{align*}
(1-p)\Delta_l = d^*/2 \quad \text{and} \quad p \Delta_h = d^*/2.
\end{align*}

The optimization problem in Case 2a reduces to (OPT-Bind):
\begin{align}
\min_{r_l\geq 0, r_h\geq 0}\quad& R_l(r) \tag{OPT-Bind}\\
\text{s.t.}\quad& (1-p)(f_l\cdot r_l - f_l\cdot r_h) = d^*/2 \notag \\
& p(f_h\cdot r_h - f_h\cdot r_l) = d^*/2 \notag
\end{align}
This LP has 2 equality constraints. The rank of the constraint matrix is at most 2. By the Fundamental Theorem of LP, there exists an optimal BFS with at most 2 positive variables. Thus, $N_r \leq 2$.

\subparagraph{Case 2b: Consistency condition fails.}
If the cost function is strictly convex such that $(\gamma^* - \frac{1}{2}) d^* > c(\gamma^*)$, the consistency condition \eqref{eq:consistency} is violated. The requirements of Case 2 (that \eqref{eq:MH}, \eqref{eq:ic1p}, and \eqref{eq:ic2p} all hold with $s_1=s_2=0$) are mathematically inconsistent. The feasible set for Case 2 is empty. The optimal solution must therefore fall into Case 1.

\medskip
Finally, combining all cases, there exists an optimal solution $(r^*, s_1^*, s_2^*)$ such that 
the total number of strictly positive payments made by the contract $r^*$ across all outcomes and reports is at most 2. That is,
$$
\left| \{ (\sigma, i) : r_{\sigma,i}^* > 0 \} \right| \le 2.
$$
Let $V^+$ be the set of distinct positive values utilized by the optimal contract $r^*$. Since there are at most 2 positive payments in total, the number of distinct positive values must also be at most~2 (i.e., $|V^+| \le 2$).
Therefore, the set of all payment values used by the contract is $V^+ \cup \{0\}$. This aligns with \cref{def:3-tier contract} of a three tier contract (which allows for at most three distinct values, including zero).

Finally, since the cost-minimizing implementation for any $\gamma^*$ is a three tier contract, the overall optimal contract for the principal must also be a three tier contract. 
\end{proof}

\begin{proof}[Proof of \cref{thm:symmetric_mlrp}]
Similar to the proof of \cref{thm:optimal_contract}, we analyze the payment minimization problem \eqref{eq:opt_contract_simplified} for implementing a target accuracy $\gamma^* > \frac{1}{2}$. (The case $\gamma^*=\frac{1}{2}$ is trivial, implemented by the zero contract). Let $d^* \triangleq c'(\gamma^*) > 0$.

\paragraph{Simplification of Optimization Program.}
We first show that it is without loss of generality to restrict attention to symmetric contracts. 
A contract $r$ is symmetric if $r_{l,i} = r_{h, n+1-i}$ for all $i$. If $r$ is an optimal contract, the symmetrized contract $r^S$, defined by $r^S_{l,i} = r^S_{h, n+1-i} = \frac{1}{2}(r_{l,i} + r_{h, n+1-i})$, is also optimal due to the linearity of \eqref{eq:opt_contract_simplified} and the symmetry of the environment.

For any symmetric contracts, let $x_i = r_{l,i}$. Then $r_{h,i} = x_{n+1-i}$. The optimization problem can then be formulated in terms of $n$ variables $x_i \ge 0$.
Moreover, under a symmetric contract in a symmetric environment, $f_l\cdot r_l = f_h\cdot r_h$ and $f_l\cdot r_h = f_h\cdot r_l$. Consequently, the constraints \eqref{eq:IC1} and \eqref{eq:IC2} become identical (denoted as (IC-S)). 
Moreover, the symmetry implies that the moral hazard constraint can be simplified to 
\begin{align}\label{eq:mhs}
(f_l - f_h)\cdot x = d^*. \tag{MH-S}
\end{align}
The objective function also simplifies, with $\hat{T}(\gamma^*; r)$ becoming:
\begin{align*}
\hat{T}(\gamma^*; x) &= (\gamma^* f_l + (1-\gamma^*) f_h) \cdot x,
\end{align*}
where recall that the prior $p=\frac{1}{2}$ in symmetric environments.

Now we re-examine the incentive constraint in symmetric environments. 
Note that the simplified (IC-S) can be represented as 
\begin{align*}
\hat{T}(\gamma^*;x) - \frac{1}{2}(f_l + f_h)\cdot x \geq c(\gamma^*).
\end{align*}
For the left-hand side, we have 
\begin{align*}
\hat{T}(\gamma^*;x) - \frac{1}{2}(f_l + f_h)\cdot x 
= (\gamma^* f_l + (1-\gamma^*) f_h) \cdot x
- \frac{1}{2}(f_l + f_h)\cdot x 
= \rbr{\gamma^*-\frac{1}{2}} \cdot \rbr{f_l-f_h}\cdot x.
\end{align*}
By \eqref{eq:mhs}, (IC-S) becomes $(\gamma^* - \frac{1}{2}) d^* \ge c(\gamma^*)$.
Since $d^* = c'(\gamma^*)$, this is equivalent to $\frac{c(\gamma^*) - c(\frac{1}{2})}{\gamma^* - \frac{1}{2}} \le c'(\gamma^*)$ (as $c(\frac{1}{2})=0$). This inequality always holds by the convexity of $c(\gamma)$ (\cref{asp:convex}).
Therefore, the constraint (IC-S) is also redundant given \eqref{eq:mhs}.

Given the simplifications, the optimization program in symmetric environments reduces to 
\begin{align}
\min_{x\geq 0}\quad& (\gamma^* f_l + (1-\gamma^*) f_h) \cdot x \tag{OPT-S}\label{eq:opt_contract_symmetric} \\
\text{s.t.}\quad& (f_l - f_h)\cdot x = d^*. \tag{MH-S}
\end{align}

\paragraph{The Shape of the Optimal Contract under MLRP.}
The optimization problem reduces to a simple LP. To minimize the cost, we utilize the variables $x_i$ that maximize the ``bang for the buck,'' defined by the ratio of the constraint coefficient (Incentive) to the cost coefficient:
$$
R_i = \frac{f_{l,i} - f_{h,i}}{\gamma^* f_{l,i} + (1-\gamma^*) f_{h,i}}.
$$
We only consider indices where $f_{l,i} > f_{h,i}$ since $d^*>0$. Let $L_i = f_{h,i}/f_{l,i}$ be the likelihood ratio ($L_i < 1$).
$$
R_i = \frac{1 - L_i}{\gamma^* + (1-\gamma^*) L_i}.
$$
Consider the function $g(L) = \frac{1-L}{\gamma^* + (1-\gamma^*)L}$. Its derivative is $g'(L) = -1/(\gamma^* + (1-\gamma^*)L)^2 < 0$. Thus, $g(L)$ is strictly decreasing in $L$.

By the MLRP Assumption (\cref{asp:mlrp}), $L_i$ is weakly increasing in $i$. Therefore, $R_i$ is weakly decreasing in~$i$. To maximize efficiency, we must prioritize the smallest indices (lowest returns).
Let $R_{\max} = \max_i R_i$. Let $K$ be the set of indices that maximize $R_i$. Since $R_i$ is weakly decreasing, $K$ must be an initial segment $K=\{1, 2, \dots, k\}$ for some $k \ge 1$. Any optimal solution $x^*$ must have $x^*_i = 0$ for $i \notin K$.

Now we construct an optimal threshold contract $x^T$. Define $x^T_i = r^*$ if $i \in K$ (i.e., $i \le k$) and $x^T_i=0$ otherwise. We select $r^*$ such that $x^T$ satisfies (MH-S):
$$
r^* \sum_{i=1}^k (f_{l,i} - f_{h,i}) = d^* \implies r^* = \frac{d^*}{\sum_{i=1}^k (f_{l,i} - f_{h,i})}.
$$
We verify the cost of $x^T$. The minimum cost of the LP is $C^* = d^*/R_{\max}$. For $i \in K$, the cost coefficient $C_i$ satisfies $C_i = (f_{l,i} - f_{h,i})/R_{\max}$.
\begin{align*}
C^T = \sum_{i=1}^k C_i r^* = \frac{r^*}{R_{\max}} \sum_{i=1}^k (f_{l,i} - f_{h,i}) = \frac{d^*}{R_{\max}} = C^*.
\end{align*}
Therefore, the threshold contract $x^T$ is optimal for \eqref{eq:opt_contract_symmetric}.

Finally, the constructed contract pays $r_{l,i}=r^*$ if $v_i \le \underline{v}=v_k$. 
By symmetry, $r_{h,i} = x_{n+1-i}$, so it pays $r_{h,i}=r^*$ if $v_i \ge \bar{v}=v_{n+1-k}$. 
This matches the definition of a threshold contract with thresholds $(\underline v,\bar v)$.

\end{proof}

\section{Additional Missing Proofs}
\label{apx:additional_proofs}
\subsection{Worst Case Gap of Linear Contracts}
\label{subapx:linear_gap}
\begin{proof}[Proof of \cref{prop:linear_gap}]
We divide the analysis into two cases. 
\begin{itemize}
\item $p\mu_h+(1-p)\mu_l > 1$. In this case, under linear contracts, the principal can at least guarantee a payoff of $p\mu_h+(1-p)\mu_l$ by offering a linear contract of $\alpha=0$ and investing in the risky option. 
Moreover, the payoff under the optimal contract is upper bounded by the welfare with full information, which is 
$p\mu_h+(1-p)$. 
The difference is at most $(1-p)(1-\mu_l)\leq 1$, 
which implies that the multiplicative gap is at most 
$1+\frac{1}{p\mu_h+(1-p)\mu_l} < 2$.

\item $p\mu_h+(1-p)\mu_l \leq 1$. In this case, under linear contracts, the principal's payoff is at least~$1$. 
Again, the payoff under the optimal contract is upper bounded by $p\mu_h+(1-p)$. 
The multiplicative gap is at most 
$p\mu_h+(1-p)\leq p\mu_h+1\leq 2$.
\end{itemize}
Combining both cases, \cref{prop:linear_gap} holds.
\end{proof}

\subsection{Optimality of Non Threshold Contracts}
\label{subapx:non-threshold}
In this section, we provide a simple example showing that, without the MLRP assumption, the optimal contract may not take the form of thresholds. 

Specifically, consider a symmetric environment where the realized return can take 5 possible values $(0,\frac{1}{2},1,\frac{3}{2},2)$. 
The probabilities of those 5 values are $(\frac{1}{8},\frac{3}{8},\frac{1}{4},\frac{1}{8},\frac{1}{8})$ given the low state $l$, and $(\frac{1}{8},\frac{1}{8},\frac{1}{4},\frac{3}{8},\frac{1}{8})$ given the high state $h$.
In this simple example, the MLRP assumption is violated. Specifically, the likelihood ratio $\frac{f_h}{f_l}$ is maximized (minimized) at returns $\frac{3}{2}$ ($\frac{1}{2}$). 
To maximize the principal's payoff, or equivalently to minimize expected payments subject to inducing the desired effort level and truthful reporting, the optimal contract places positive payments only on two report return cells, $(\sigma_h, v=3/2)$ and $(\sigma_l, v=1/2)$, and sets all other payments to zero. This is not representable by a threshold contract.

\subsection{Partially Observable Returns}
\label{subapx:observe}
Our paper focuses on the model where the realized return of the risky investment is observable regardless of the investment decision. This is plausible in many applications. On the other hand, there also exist applications where the return of the risky investment cannot be observed if the principal chooses the safe investment. We refer to such environments as partially observable environments. In this section, we extend our results to these environments by showing that our main characterizations generalize with negligible loss. The key idea is to occasionally choose the risky investment for observability and scale payments in those events, which preserves incentives in expectation. A caveat is that this construction scales realized bonuses by a factor $1/\epsilon$, so it relies on there being no tight cap on realized bonuses (or on caps being sufficiently large).

%Our paper focuses on the model where the realized return of the risky investment is observable regardless of the investment decision. This is plausible in many applications. On the other hand, there also exist applications where the return of the risky investment cannot be observed if the principal chooses the safe investment. We refer to such environments as partially observable environments. The previous results in this paper are silent regarding those environments. 
%In this section, we will provide extensions to those environments by showing that our main results can be generalized with negligible loss. 
%A small remark regarding our extension result is that it can scale up the maximum possible bonus for the agent, so the extension relies on the fact that there is no tight cap on realized bonuses.

Let $\opt$ be the optimal net return when the realized returns are always observable. 
\begin{proposition}\label{prop:partial}
In partially observable environments, for any $\epsilon>0$, there exists a three tier contract such that the expected net return is at least $\opt-\epsilon$. 
\end{proposition}
\begin{proof}
For any $\epsilon>0$, consider the decision rule where with probability $\epsilon$, the principal always chooses the risky investment regardless of the posterior belief reported by the expert. 
With the remaining probability $1-\epsilon$, the principal makes the optimal investment decision based on the posterior belief reported by the expert. 
Letting contract $S$ be the three tier contract that is optimal for the principal in perfectly observable environments, 
consider another contract $\hat{S}$ such that $\hat{S}$ only rewards the expert in the $\epsilon$ probability event where the risky investment is chosen regardless of the reports from the expert.
Moreover, for any report $\sigma$ and realized return $v$, 
we have $\hat{S}(\sigma,v) = \frac{1}{\epsilon}\cdot S(\sigma,v)$.
We remark that for contract $\hat{S}$, in the event that the principal chooses the investment based on the report of the agent, which occurs with probability $1-\epsilon$, the expert always receives a zero reward even when the expert recommends the risky investment in this event. 

In this construction, the expected rewards of the expert for any choice of accuracy $\gamma$ are the same given either contract $S$ in perfectly observable environments or contract $\hat{S}$ in partially observable environments, as the rewards occur with probability $\epsilon$ but are scaled by a multiplicative factor of $\frac{1}{\epsilon}$ in partially observable environments.
Therefore, the equilibrium accuracy level chosen by the expert and the expected equilibrium payment to the expert remain the same. 

Finally, since the principal chooses the risky investment with probability $\epsilon$ regardless of the report of the expert, this leads to a loss in expected returns for the principal. 
Note that in these events, the principal suffers a loss in risky investment only when the optimal investment is the safe investment. The loss conditional on that event is at most 1. 
Therefore, the expected loss in investment returns for the principal is at most $\epsilon$. 
Combining the observations, the expected loss in net returns is at most $\epsilon$. 
\end{proof}
Note that in \cref{prop:partial}, the principal can flexibly design the parameter $\epsilon>0$. 
With the choice of a smaller parameter $\epsilon$, the loss in the partially observable environments is smaller. 
The only caveat is that in the construction, we may need to scale up the realized reward of the agent by a multiplicative factor of $\frac{1}{\epsilon}$. This could be undesirable for the principal if they have a limited capacity for rewarding the expert. In that case, the principal needs to consider a more careful tradeoff between the realized rewards and the expected net returns. The details of those tradeoffs are beyond the scope of the current paper. 

Finally, for symmetric and MLRP environments, we can generalize the results of threshold contracts with negligible loss as well. The details for the latter extension are omitted here to avoid repetition.

% simple examples illustrating the complications when the return of the risky investment can be unobservable. 

% One crucial challenge in this alternative model, which is missing in our paper, is that the investment decision now interacts with the expert's incentives for information acquisition. 
% Specifically, in our paper, since the return of the risky investment is always observed by the principal, the investment decision does not affect the principal's ability to use the observed returns to reward the experts for providing precise information. Therefore, in our model, it is always without loss to make the investment decision that is optimal given the true posterior belief of the expert after acquiring information. 
% However, this is not true when the return of the risky investment can be unobservable. 

\end{document}